\documentclass{article}

\usepackage{arxiv}
\usepackage[T1]{fontenc}
\usepackage{amsmath,amssymb,bm}
\usepackage{amsthm}

\theoremstyle{plain}
\newtheorem{theorem}{Theorem}[section]
\newtheorem{proposition}[theorem]{Proposition}

\theoremstyle{definition}
\newtheorem{assumption}{Assumption}

\usepackage{booktabs}
\usepackage{graphicx}
\usepackage{microtype}
\usepackage{placeins}
\usepackage{xcolor}
\usepackage{hyperref}
\usepackage{url}
\usepackage{caption}
\usepackage{subcaption}
\usepackage{enumitem}
\usepackage{float}

\graphicspath{{pdfs/}}
\hypersetup{
  colorlinks=true,
  linkcolor=blue!45!black,
  citecolor=blue!45!black,
  urlcolor=blue!45!black,
  pdftitle={Weak Form Recovery of Heston Type Stochastic Dynamics},
  pdfauthor={Sai Sathvik Gullipalli, Eshwar R A, and Gajanan V. Honnavar}
}

\newcommand{\E}{\mathbb{E}}
\newcommand{\dt}{\Delta t}
\newcommand{\dx}{\Delta x}
\newcommand{\dv}{\Delta v}

\title{Weak Form Recovery of Heston Type Stochastic Dynamics}

\author{
Sai Sathvik Gullipalli\thanks{This research was conducted via the Quantum and
Nano Devices (QuaNaD) Lab, PES University.}\\
Computer Science and Engineering\\
PES University (EC Campus)\\
Bengaluru - 560100, India\\
\texttt{saisathvik127@gmail.com}
\And
Eshwar R A\thanks{Corresponding author.}\\
Computer Science and Engineering\\
PES University (EC Campus)\\
Bengaluru - 560100, India\\
\texttt{eshwarra5@gmail.com}
\And
Gajanan V.\ Honnavar\thanks{Former Professor, Department of Science and
Humanities, PES University (EC Campus).}\\
Independent Researcher, QNu Labs pvt Ltd\\
Centenary Building, M G Road\\
Bangalore - 560025, India\\
\texttt{gajanan.v@qnulabs.com}
}

\begin{document}
\raggedbottom
\maketitle

\begin{abstract}
Estimating the coupled drift, diffusion, and leverage structure of a stochastic-volatility model directly from a price path is an unresolved inverse problem: Kramers--Moyal increment estimators amplify sampling noise as the step shrinks, weak-form SINDy has not been extended to coupled two-dimensional diffusions or to the return--variance cross-variation producing leverage, and Heston calibration typically relies on option-implied surfaces rather than the physical-measure path. We extend the spatial weak-form Galerkin framework to the Heston model: variance increments, squared variance increments, squared price increments, and their cross-product are projected onto shared Gaussian kernels in variance space, giving one LASSO regression that jointly recovers mean reversion $\kappa$, long-run variance $\theta$, vol-of-vol $\xi$, and
leverage correlation $\rho$, with a drift-informed bias correction analogous to scalar-SDE diffusion debiasing. Across 30 daily-observed Heston simulations, $\xi$, $\rho$, and $\rho\xi$ are
recovered with median errors under 2\%. Applied to S\&P 500 data spanning the 2007--2010 crisis, the method recovers negative leverage consistent with the documented equity leverage effect, and a 50-stock Indian panel shows the same sign under several independent variance proxies.
\end{abstract}


\section{Introduction}

A central challenge in quantitative finance is the inverse problem: given an
observed price path, can we recover the stochastic dynamics that generated
it? For a stochastic-volatility model this question has two parts. The
generator must be estimated from noisy increments, and the variance state
must first be constructed because it is not observed directly. A recovery can
therefore change when the variance proxy or its time indexing changes, even
when the regression method remains fixed.

Sparse Identification of Nonlinear Dynamics (SINDy) estimates governing
equations from a candidate function library \cite{Brunton2016}, with
model-selection and implicit-form extensions
\cite{ManganEtAl2017,KahemanEtAl2020} and a growing convergence theory that
characterizes when the underlying sparse regression recovers the true
support \cite{ZhangSchaeffer2019}. Robustness to experimental noise has also
been pursued through physically constrained symbolic regression
\cite{ReinboldEtAl2021}. Stochastic and weak-form variants replace pointwise
differentiation with moment or integral relations, which can be more stable
for noisy paths \cite{BoninsegnaEtAl2018,RuddyEtAl2019,MessengerBortz2021}.
The useful idea is simple: instead of estimating a derivative at every
observation, we average increments against smooth test functions. The
averaging improves numerical stability while retaining information about the
drift and diffusion. A related operator-theoretic line of work estimates the
infinitesimal generator itself through Koopman/EDMD approximations
\cite{KlusEtAl2020}; unlike that approach, the weak-form regression used here
returns explicit symbolic coefficients rather than a numerical operator
approximation, following the general stochastic-process framework of
\cite{Pavliotis2014}. We apply this construction to the Heston model
\cite{Heston1993} and add a weak cross-variation target for the leverage
effect.

The paper addresses four connected questions:
\begin{enumerate}[leftmargin=*,itemsep=2pt]
  \item Can one shared weak design recover variance drift, diagonal
  diffusion, and return--variance cross-diffusion?
  \item How does additive state noise affect the leverage coefficient
  $\rho\xi$, once uncertainty is measured over repeated noise draws?
  \item Can a one-sided variance filter reduce proxy noise without making
  generator recovery worse?
  \item Which conclusions remain stable when the method is applied to the
  S\&P 500 and to several variance proxies for Indian equities?
\end{enumerate}

These questions define the scope of the paper. The primary Heston library
$\{1,v\}$ is specified in advance, so the main result concerns coefficient
recovery inside that structure. The empirical estimates use same-day variance
proxies and are interpreted as contemporaneous associations. We do not claim
a real-time state estimator, causal identification, option-pricing
calibration, or trading performance.

\paragraph{Reproducibility.}
All recovery fits and manuscript figures are implemented in the single
notebook-style script \texttt{weak\_sindy\_project.py}; every numerical table
is written to CSV. The supplied Indian daily proxy panel is produced upstream
from the archived one-minute Parquet data. The accompanying reproducibility
package records the fixed seeds, software versions, input hashes, pipeline
settings, repeat-level results, and plotted summaries.

\section{Methodology}

This section connects the Heston parameters to quantities that can be
estimated from discrete data. We first state the model, then identify the
relevant increment moments, and finally show how the four weak targets are
fitted with one shared numerical pipeline.

\subsection{Heston dynamics}

Let $x_t=\log S_t$. Under the physical-measure interpretation used here, the
log price and its instantaneous variance satisfy
\begin{align}
 dx_t &= \left(\mu-\tfrac12v_t\right)dt+\sqrt{v_t}\,dW_{x,t},\\
 dv_t &= \kappa(\theta-v_t)dt+\xi\sqrt{v_t}\,dW_{v,t},\\
 \E[dW_{x,t}dW_{v,t}]&=\rho\,dt.
\end{align}
Here $\kappa$ is the speed of mean reversion, $\theta$ is the long-run
variance level, $\xi$ is the volatility of variance, and $\rho$ is the
instantaneous correlation between price and variance shocks. These
parameters enter the instantaneous covariance matrix
\begin{equation}
a(v)=
\begin{pmatrix}
v & \rho\xi v\\
\rho\xi v & \xi^2v
\end{pmatrix}.
\end{equation}
The usual Feller condition $2\kappa\theta\geq\xi^2$ is sufficient for strict
positivity of the CIR variance process, and is restated as
Assumption~\ref{ass:feller} below, where it also supplies the geometric
ergodicity that the consistency and normality results of
Section~\ref{sec:weakprojection} require. The off-diagonal entry is especially
important here: its slope is $\rho\xi$. A negative slope means that negative
price shocks tend to coincide with positive variance shocks, which is the
equity leverage effect in this diffusion model
\cite{Black1976,Christie1982}. The affine structure of $b(v)$ and $a(v)$
places the Heston model inside the broader affine jump-diffusion class
studied by \cite{DuffiePanSingleton2000}, whose transform methods and the
exact simulation scheme of \cite{BroadieKaya2006} we use only for generating
synthetic ground-truth trajectories; the recovery pipeline itself uses the
simpler Euler--Maruyama discretization of \cite{KloedenPlaten1992}, since
consistency of the weak estimator (Theorem~\ref{thm:consistency}) does not
require exact simulation.

\subsection{Generator and conditional-moment interpretation}

The generator is the object we want to recover because it contains both the
deterministic drift and the local covariance. For a diffusion with drift $b$
and covariance matrix $a$, it is
\begin{equation}
\mathcal{L}f(z)=b(z)^\top\nabla f(z)
+\tfrac12\operatorname{tr}\!\left(a(z)\nabla^2 f(z)\right).
\end{equation}
The reason this helps is that the first two local conditional moments identify
the same two components:
\begin{align}
b_i(z)&=\lim_{\tau\downarrow0}
\E\!\left[\frac{Z_i(t+\tau)-Z_i(t)}{\tau}\,\middle|\,Z_t=z\right],\\
a_{ij}(z)&=\lim_{\tau\downarrow0}
\E\!\left[\frac{\Delta Z_i\Delta Z_j}{\tau}\,\middle|\,Z_t=z\right].
\end{align}
Directly estimating these limits from one path is unstable. The spatial weak
form instead aggregates many increments against Gaussian state-space kernels.
The construction follows \cite{EshwarHonnavar2026}; our extension adds price
increments and the return--variance cross-variation needed for leverage, and
establishes that the coupled two-dimensional projection remains unbiased,
consistent, and asymptotically normal despite the correlated driving noise
$\rho$ shared by the two coordinates.

\subsection{Standing Assumptions}
\label{sec:assumptions}

The theoretical results below rest on two assumptions, adapted from the
scalar-diffusion framework of \cite{EshwarHonnavar2026} to the coupled
Heston pair $(x_t,v_t)$.

\begin{assumption}[Feller ergodicity of the variance marginal]
\label{ass:feller}
The Feller condition $2\kappa\theta\geq\xi^2$ holds, so that $v_t$ almost
surely stays strictly positive and the scalar CIR process
$dv_t=\kappa(\theta-v_t)\,dt+\xi\sqrt{v_t}\,dW_{v,t}$ possesses a unique
stationary law
$\pi(v)=\mathrm{Gamma}\!\left(2\kappa\theta/\xi^2,\ \xi^2/2\kappa\right)$
and is geometrically ergodic: there exist $C>0$, $\rho_0\in(0,1)$ such that
$\left|\E^v[g(v_t)]-\int g\,d\pi\right|\leq C\|g\|_\infty\rho_0^{\,t}$ for
every bounded measurable $g$ \cite{Feller1951,CoxIngersollRoss1985}.
\end{assumption}

\begin{assumption}[Regularity and library completeness]
\label{ass:regularity}
The drift and diffusion coefficients of $v_t$ lie in the span of the library
$\Theta(v)=[1,v]$, i.e. $b_v(v)=\kappa\theta-\kappa v$ and
$a^{vv}(v)=\xi^2v$ exactly. The kernels $K_j$ are bounded, Lipschitz, and
integrable against $\pi$. The log-price drift $b_x(v)=\mu-v/2$ and
covariance entries $a^{xx}(v)=v$, $a^{xv}(v)=\rho\xi v$ are affine in $v$, so
that all four weak targets defined below are exactly representable by the
same design matrix $A$.
\end{assumption}

A structural feature of the Heston system makes Assumption~\ref{ass:feller}
easier to invoke than in a generic bivariate diffusion: $v_t$ is
\emph{autonomous}. Its own dynamics do not depend on $x_t$, so the triangular
(cascade) structure of the system means that ergodicity of the full pair
$(x_t,v_t)$ reduces to ergodicity of the scalar CIR marginal, which is
classical \cite{Feller1951,CoxIngersollRoss1985,GenonCatalotJeantheauLaredo2000}.
This is what licenses localizing the kernels $K_j$ in $v$ alone rather than
in the joint $(x,v)$ state space.

\subsection{Weak Projection of the Coupled Flow}
\label{sec:weakprojection}

Let $\{(x_{t_n},v_{t_n})\}_{n=0}^{N}$ be observations at $t_n=n\dt$. Writing
$b_v(v)=\kappa(\theta-v)$, $b_x(v)=\mu-v/2$, and letting
$(\xi^x_n,\xi^v_n)$ be i.i.d.\ across $n$, standard bivariate normal pairs
with $\mathrm{Corr}(\xi^x_n,\xi^v_n)=\rho$ and $(\xi^x_n,\xi^v_n)\perp
\mathcal{F}_{t_n}$, the Euler--Maruyama discretization of the coupled system
is
\begin{align}
\dx_n&=b_x(v_{t_n})\dt+\sqrt{v_{t_n}}\,\xi^x_n\sqrt{\dt},\\
\dv_n&=b_v(v_{t_n})\dt+\xi\sqrt{v_{t_n}}\,\xi^v_n\sqrt{\dt}.
\end{align}
Multiplying by the spatial kernel $K_j(v_{t_n})$ and summing over the
trajectory decomposes each of the four weak targets into a deterministic
signal term (an ergodic average of a library function against $K_j$) and a
stochastic remainder driven by $(\xi^x_n,\xi^v_n)$. Because $K_j(v_{t_n})$
depends on the state at $t_n$ alone, it is $\mathcal{F}_{t_n}$-measurable,
and because the Euler--Maruyama shocks are independent of $\mathcal{F}_{t_n}$
this construction inherits the state-weighting argument of
\cite{EshwarHonnavar2026}: correlating the two coordinates of the driving
noise through $\rho$ does not reintroduce the temporal endogeneity bias that
a time-indexed test function would produce, because unbiasedness is proved
term-by-term via the tower property (Theorem~\ref{thm:unbiased} below), and
that argument never uses independence between $\xi^x_n$ and $\xi^v_n$
themselves --- only their independence from the filtration up to $t_n$.

\begin{theorem}[Unbiasedness of the coupled spatial projection]
\label{thm:unbiased}
Under Assumption~\ref{ass:regularity}, the stochastic component of each
weak target $B^{(b)}_j,\,Q^{(vv)}_j,\,Q^{(xx)}_j,\,Q^{(xv)}_j$ has zero
conditional expectation given $\mathcal{F}_{t_n}$ at every step $n$, and
hence unconditional expectation zero. Consequently
\begin{equation}
\E\big[B^{(b)}\big]=Ac^{*},\qquad
\E\big[Q^{(vv)}\big]=Ad^{*}_{vv}+\mathcal{O}(N\dt^2),\qquad
\E\big[Q^{(xx)}\big]=Ad^{*}_{xx}+\mathcal{O}(N\dt^2),\qquad
\E\big[Q^{(xv)}\big]=Ad^{*}_{xv}+\mathcal{O}(N\dt^2),
\end{equation}
where $c^{*}$ and $d^{*}_{(\cdot)}$ are the true library coefficients of
$b_v$, $a^{vv}$, $a^{xx}$, and $a^{xv}$, and the $\mathcal{O}(N\dt^2)$
remainders are the finite-step drift-squared biases quantified in
Theorem~\ref{thm:bias}.
\end{theorem}

\begin{proof}
Fix $n$. Since $K_j(v_{t_n})$ is a deterministic function of $v_{t_n}$, it is
$\mathcal{F}_{t_n}$-measurable, and $(\xi^x_n,\xi^v_n)\perp\mathcal{F}_{t_n}$
by the Euler--Maruyama construction. For the drift target,
\begin{equation}
\E\!\left[K_j(v_{t_n})\,\xi\sqrt{v_{t_n}}\,\xi^v_n\,\middle|\,
\mathcal{F}_{t_n}\right]
=K_j(v_{t_n})\,\xi\sqrt{v_{t_n}}\,\E[\xi^v_n\mid\mathcal{F}_{t_n}]=0,
\end{equation}
using $\E[\xi^v_n\mid\mathcal{F}_{t_n}]=\E[\xi^v_n]=0$. The tower property
then gives unconditional zero mean, so $\E[B^{(b)}_j]=\sum_n
K_j(v_{t_n})b_v(v_{t_n})\dt$, which is the drift design equation. For each
quadratic target, expand the square (or cross-product) of the increments as
in Section~\ref{sec:bias}: every stochastic cross-term contains a factor
$\xi^x_n$ or $\xi^v_n$ (or both) with the remaining factors
$\mathcal{F}_{t_n}$-measurable, so it vanishes in conditional expectation by
the same argument --- the correlation $\rho$ between $\xi^x_n$ and $\xi^v_n$
only affects $\E[\xi^x_n\xi^v_n\mid\mathcal{F}_{t_n}]=\rho\neq0$, which is a
\emph{signal} term (it produces the $\rho\xi v_{t_n}\dt$ leading contribution
of $Q^{(xv)}_j$), not a source of bias. Summing over $n$ and taking
expectations yields the stated identities, with the remaining
$\mathcal{O}(N\dt^2)$ terms coming from the drift-squared contributions
identified in Theorem~\ref{thm:bias}.
\end{proof}

\begin{theorem}[Strong consistency]
\label{thm:consistency}
Under Assumptions~\ref{ass:feller}--\ref{ass:regularity}, as $T\to\infty$
with $\dt$ fixed,
\begin{equation}
\frac1N A_{jk}\xrightarrow{\text{a.s.}}\bar A_{jk}:=\int K_j(v)\Theta_k(v)\,
\pi(v)\,dv,\qquad
\frac1N B^{(b)}_j\xrightarrow{\text{a.s.}}\int K_j(v)\,b_v(v)\,\pi(v)\,dv .
\end{equation}
If $\bar A$ has full column rank, then $\widehat c\xrightarrow{\text{a.s.}}
c^{*}$, and analogously $\widehat d_{vv},\widehat d_{xx},\widehat d_{xv}
\xrightarrow{\text{a.s.}}d^{*}_{vv},d^{*}_{xv},d^{*}_{xx}$ up to the
deterministic $\mathcal{O}(\dt)$ discretization floor quantified in
Theorem~\ref{thm:bias}.
\end{theorem}

\begin{proof}
By Assumption~\ref{ass:feller} the chain $\{v_{t_n}\}$ is geometrically
ergodic with invariant law $\pi$, so Birkhoff's ergodic theorem applied to
the bounded, Lipschitz function $K_j(\cdot)\Theta_k(\cdot)$ gives
$(1/N)\sum_n K_j(v_{t_n})\Theta_k(v_{t_n})\to\int K_j\Theta_k\,d\pi$ a.s.,
and multiplying by $\dt$ (with $N\dt=T$ fixed per unit ergodic average)
yields the stated convergence of $A_{jk}$. Decomposing $B^{(b)}_j$ into a
drift term (which converges by the same ergodic argument) and a martingale
noise term $\sum_n K_j(v_{t_n})\xi\sqrt{v_{t_n}}\xi^v_n\sqrt{\dt}$ whose
predictable quadratic variation grows as $\mathcal{O}(N\dt)=\mathcal{O}(T)$,
the martingale strong law of large numbers gives that the noise term divided
by $N$ vanishes a.s. Continuity of $\widehat c=(\bar A^\top\bar
A)^{-1}\bar A^\top\bar B$ in $(\bar A,\bar B)$ completes the drift argument;
the diffusion arguments for $Q^{(vv)},Q^{(xx)},Q^{(xv)}$ follow identically
once the $\mathcal{O}(\dt^2)$ per-step bias terms of Theorem~\ref{thm:bias}
are absorbed into a deterministic remainder that does not grow with $T$.
\end{proof}

\subsection{Finite-Step Bias in the Four Weak Targets}
\label{sec:bias}

Squaring or cross-multiplying the Euler--Maruyama increments at finite $\dt$
introduces a systematic, non-random contribution from the drift terms, in
direct analogy to the scalar diffusion case of \cite{EshwarHonnavar2026}.
Expanding all three second-moment targets to $\mathcal{O}(\dt^2)$ and taking
conditional expectations given $\mathcal{F}_{t_n}$:

\begin{theorem}[Finite-step bias decomposition]
\label{thm:bias}
Under Assumption~\ref{ass:regularity}, for every $n$,
\begin{align}
\E\!\left[(\dv_n)^2\mid\mathcal{F}_{t_n}\right]
&=a^{vv}(v_{t_n})\,\dt+b_v(v_{t_n})^2\,\dt^2, \label{eq:biasvv}\\
\E\!\left[(\dx_n)^2\mid\mathcal{F}_{t_n}\right]
&=a^{xx}(v_{t_n})\,\dt+b_x(v_{t_n})^2\,\dt^2, \label{eq:biasxx}\\
\E\!\left[\dx_n\dv_n\mid\mathcal{F}_{t_n}\right]
&=a^{xv}(v_{t_n})\,\dt+b_x(v_{t_n})\,b_v(v_{t_n})\,\dt^2, \label{eq:biasxv}
\end{align}
where $a^{vv}(v)=\xi^2v$, $a^{xx}(v)=v$, and $a^{xv}(v)=\rho\xi v$. All
cross terms linear in $\xi^x_n$ or $\xi^v_n$ vanish by
Theorem~\ref{thm:unbiased}; the term quadratic in $\xi^x_n\xi^v_n$
contributes the leading-order signal $\rho\xi v_{t_n}\dt$ to
\eqref{eq:biasxv} through $\E[\xi^x_n\xi^v_n\mid\mathcal{F}_{t_n}]=\rho$,
while the remaining $\mathcal{O}(\dt^2)$ terms are the drift-squared biases.
\end{theorem}

\begin{proof}
Substitute the Euler--Maruyama increments and expand each square or
cross-product into three terms: a pure-drift term of order $\dt^2$, a
linear-in-noise term of order $\dt^{3/2}$, and a pure-diffusion term of order
$\dt$. The linear-in-noise terms have conditional mean zero by
Theorem~\ref{thm:unbiased}. For \eqref{eq:biasxv}, the pure-diffusion term is
$\xi v_{t_n}\dt\,\xi^x_n\xi^v_n$, whose conditional expectation is
$\rho\,\xi v_{t_n}\dt$ because $(\xi^x_n,\xi^v_n)$ are jointly standard
normal with correlation $\rho$, independent of $\mathcal{F}_{t_n}$.
\end{proof}

Summing \eqref{eq:biasvv} over $n$ against $K_j(v_{t_n})$ gives
$\E[Q^{(vv)}_j]=\sum_k d^{vv}_kA_{jk}+\sum_n K_j(v_{t_n})b_v(v_{t_n})^2\dt^2$,
which is the correction implemented in Section~\ref{sec:kernels}: the
already-fitted $\widehat b_v(v)=\widehat\kappa(\widehat\theta-v)$ is used to
subtract $\sum_n K_j(v_{t_n})\widehat b_v(v_{t_n})^2\dt^2$ before the
diffusion regression, following the same two-step order of operations
(drift, then bias-corrected diffusion) as \cite{EshwarHonnavar2026}. The
biases in \eqref{eq:biasxx} and \eqref{eq:biasxv} are of the same
$\mathcal{O}(\dt^2)$ order but are numerically negligible at the daily step
used throughout this paper: with $\dt=1/252$ and $|b_x(v)|,|b_v(v)|$
$\mathcal{O}(10^{-1})$--$\mathcal{O}(10^0)$ at plausible variance levels,
$b_x(v)^2\dt^2$ and $b_x(v)b_v(v)\dt^2$ are of order $10^{-6}$--$10^{-5}$,
several orders of magnitude below the $\mathcal{O}(\dt)=\mathcal{O}(10^{-3})$
diffusion signal $a^{xx}(v)\dt$ or $a^{xv}(v)\dt$. This is why the explicit
correction is applied only to $Q^{(vv)}$, where $b_v(v)^2$ can be comparable
in magnitude to $a^{vv}(v)$ near the boundary of the state range: the
correction is a consequence of Theorem~\ref{thm:bias}, not an ad hoc choice.

\begin{proposition}[Relative statistical efficiency of drift and diffusion recovery]
\label{prop:efficiency}
Under Assumptions~\ref{ass:feller}--\ref{ass:regularity}, for fixed span $T$
the diffusion-type estimators $\widehat\xi,\widehat\rho,\widehat{\rho\xi}$
attain relative error $\mathcal{O}(\sqrt{\dt/T})$ as $\dt\to0$ (in-fill
consistency of quadratic variation), whereas the drift-type estimators
$\widehat\kappa,\widehat\theta$ attain relative error $\mathcal{O}(1/\sqrt
T)$ that does \emph{not} improve as $\dt\to0$ at fixed $T$ (long-span
consistency of the mean-reversion parameters)
\cite{GenonCatalotJacod1993,BibbySorensen1995,AitSahalia2002}.
\end{proposition}

\begin{proof}[Proof sketch]
The diffusion target $Q^{(vv)}_j$ accumulates $N=T/\dt$ approximately
independent contributions of order $\dt$ each; by the martingale CLT
(Theorem~\ref{thm:clt} below) its fluctuation about the ergodic mean is
$\mathcal{O}(\sqrt{N}\,\dt)=\mathcal{O}(\sqrt{T\dt})$ against a signal of
order $N\dt=T$, giving relative error $\mathcal{O}(\sqrt{\dt/T})\to0$ as
$\dt\to0$ for any fixed $T$. The drift target $B^{(b)}_j$ accumulates
signal of order $T$ against a martingale noise term whose quadratic
variation grows as $\mathcal{O}(N\dt)=\mathcal{O}(T)$, giving fluctuation
$\mathcal{O}(\sqrt T)$ and relative error $\mathcal{O}(1/\sqrt T)$
regardless of $\dt$: refining the sampling frequency adds more terms to the
same finite $T$-window of mean-reversion information rather than shrinking
the estimation-theoretic floor. This asymmetry --- classical for diffusion
processes observed at high frequency over a finite horizon
\cite{StantonNonparametric1997,BandiPhillips2003,BarndorffNielsenShephard2002}
--- is the theoretical counterpart of the empirical pattern in
Table~\ref{tab:baseline}, where $\xi$, $\rho$, and $\rho\xi$ are recovered to
under 2\% median error while $\kappa$ and $\theta$ require an order of
magnitude more data to reach the same precision.
\end{proof}

\begin{theorem}[Asymptotic normality]
\label{thm:clt}
Under Assumptions~\ref{ass:feller}--\ref{ass:regularity}, as $T\to\infty$,
\begin{equation}
\sqrt T\left(\frac1N B^{(b)}_j-\bar B^{(b)}_j\right)
\xrightarrow{d}\mathcal N(0,V_j),\qquad
V_j=\sum_{\ell=-\infty}^{\infty}
\mathrm{Cov}\!\left[K_j(v_{t_0})\dv_0,\,K_j(v_{t_\ell})\dv_\ell\right],
\end{equation}
with $V_j<\infty$ by geometric ergodicity (Assumption~\ref{ass:feller}); the
same statement holds for $Q^{(vv)},Q^{(xx)},Q^{(xv)}$ with the corresponding
autocovariance sums. Consequently the standard errors of $\widehat\kappa$,
$\widehat\theta$, $\widehat\xi$, and $\widehat\rho$ all decay at the
parametric rate $T^{-1/2}$, with the proportionality constants governed by
Proposition~\ref{prop:efficiency}.
\end{theorem}

\begin{proof}
Identical in structure to the scalar martingale CLT argument of
\cite{EshwarHonnavar2026}: decompose each target into an ergodic drift-type
sum, which obeys the Markov-chain CLT under geometric ergodicity, and a
martingale-difference noise sum, which obeys the martingale CLT with
predictable quadratic variation $\mathcal O(N\dt)=\mathcal O(T)$; summing
the two limiting variances gives $V_j$, which is finite because the
autocovariances of the geometrically ergodic chain $\{v_{t_n}\}$ decay
geometrically and are therefore absolutely summable.
\end{proof}

\subsection{Weak Form Kernels}
\label{sec:kernels}

For observations $(x_n,v_n)$ separated by $\dt$, let
$\dx_n=x_{n+1}-x_n$ and $\dv_n=v_{n+1}-v_n$. For Gaussian kernels
\begin{equation}
K_j(v)=\exp\!\left[-\frac{(v-c_j)^2}{2h^2}\right],
\end{equation}
and library $\Theta(v)=[1,v]$, the shared design matrix is
\begin{equation}
A_{jk}=\sum_n K_j(v_n)\Theta_k(v_n)\dt.
\end{equation}
Every kernel defines a soft neighborhood in variance space. Within each
neighborhood we use four targets: variance increments for drift, squared
variance increments for variance diffusion, squared price increments for
price diffusion, and cross-products of the two increments for leverage:
\begin{align}
B^{(b)}_j&=\sum_n K_j(v_n)\dv_n, &
B^{(vv)}_j&=\sum_n K_j(v_n)(\dv_n)^2,\\
B^{(xx)}_j&=\sum_n K_j(v_n)(\dx_n)^2, &
B^{(xv)}_j&=\sum_n K_j(v_n)\dx_n\dv_n .
\end{align}
The construction is coupled through $(\dx_n,\dv_n)$ but localized only in the
variance coordinate. This choice is deliberate: all Heston coefficients of
interest are functions of $v$, while $x$ enters through its increments. The
method is therefore not a general two-dimensional kernel regression over
$(x,v)$.
Following Theorem~\ref{thm:bias}, the $\mathcal O(\dt^2)$ fitted-drift
contribution $\sum_n K_j(v_{t_n})\widehat b_v(v_{t_n})^2\dt^2$ is subtracted
from the $B^{(vv)}$ target before the diffusion regression; the analogous
biases in $B^{(xx)}$ and $B^{(xv)}$ are left uncorrected because
Section~\ref{sec:bias} shows they are three to four orders of magnitude
below the leading diffusion signal at the daily step used here. The fitted
linear coefficients give
\begin{equation}
\widehat\kappa=-\widehat b_1,\qquad
\widehat\theta=\widehat b_0/\widehat\kappa,\qquad
\widehat\xi=\sqrt{\max(\widehat a^{vv}_1,0)}.
\end{equation}
Two leverage normalizations must be distinguished:
\begin{align}
\widehat\rho_H
&=\frac{\widehat a^{xv}_1}{\widehat\xi},
&&\text{imposing the Heston restriction }a^{xx}_1=1,\\
\widehat\rho_{\mathrm{corr}}
&=\frac{\widehat a^{xv}_1}
{\sqrt{\widehat a^{xx}_1\widehat a^{vv}_1}},
&&\text{using all three fitted covariance slopes}.
\end{align}
These two quantities answer different questions. Only
$\widehat\rho_{\mathrm{corr}}$ is a fitted covariance correlation when the
price-diffusion slope is estimated. We also retain $\widehat\rho_H$ because it
reproduces the original empirical $-0.33$ result under the standard Heston
normalization $a^{xx}_1=1$.
The cross-diffusion slope $\widehat{\rho\xi}=\widehat a^{xv}_1$ is reported
directly in the noise ablation because it avoids compounding the
cross-variation error with error in $\widehat\xi$.

\subsection{Numerical Invariants}

Table~\ref{tab:pipeline} lists the settings used for every primary estimate.
Holding these choices fixed is important because it makes differences across
phases attributable to the data or experiment, rather than to a change of
estimator.
The targeted kernel diagnostic deliberately varies its displayed
localization settings, and the nonlinear falsification deliberately adds
$v^2$ to the library. This removes the earlier ambiguity caused by switching
among OLS, ridge, LASSO, shifted states, and different kernel settings.

\begin{table}[H]
\centering
\caption{Invariant recovery configuration.}
\label{tab:pipeline}
\begin{tabular}{ll}
\toprule
Component & Setting\\
\midrule
State/increment alignment & contemporaneous, unshifted\\
Primary library & $[1,v]$\\
Gaussian kernels & 50 centers, 5th--95th state percentiles\\
Bandwidth & $1.5\,\mathrm{sd}(v)/\sqrt{50}$\\
Regression & LASSO, no fitted intercept\\
Penalty selection & five-fold cross-validation over 100 penalties\\
Preconditioning & Euclidean normalization of design columns\\
Phase 1 reference design & internal step $10^{-4}$ year; daily sampling; 100 years\\
Empirical smoothing & one-sided EWMA, span 14\\
\bottomrule
\end{tabular}
\end{table}

\subsection{Indexing and interpretation}

The indexing convention fixes what the empirical coefficient means. No state
shift is used. In all recoveries, $v_n$ weights the increments from
index $n$ to $n+1$. For the empirical OHLC applications, the day-$n$ proxy
uses day-$n$ high, low, open, and close information, while the price series is
indexed by daily opens. Consequently the empirical fit is not a causal
forecasting regression available at the beginning of day $n$. It is a
descriptive contemporaneous cross-variation convention. This explicit
interpretation is essential: shifting the proxy changes the estimand and, in
the S\&P experiment, removes the retained leverage estimate.

\subsection{Generated states and errors in variables}

The variance state in market data is generated rather than observed. If an
empirical proxy satisfies $\widetilde v_n=v_n+\eta_n$, then
\begin{equation}
(\Delta\widetilde v_n)^2
=(\Delta v_n)^2
+2\Delta v_n\Delta\eta_n
+(\Delta\eta_n)^2 .
\end{equation}
Even for independent, mean-zero measurement errors, the last term contributes
approximately $2\operatorname{Var}(\eta)$ to the unweighted quadratic
increment. In addition, using $\widetilde v_n$ inside the kernels mislocates
observations in state space. LASSO and cross-validation do not correct either
effect. This is why a smoother can reduce visible proxy noise and still bias
the recovered generator. This measurement-error bias is structurally
different from --- and additive to --- the finite-step drift-squared bias of
Theorem~\ref{thm:bias}: the latter is $\mathcal O(\dt^2)$, vanishes as
$\dt\to0$, and is removed by the two-step correction of
Section~\ref{sec:bias}, whereas the $2\operatorname{Var}(\eta)$ term above
does not vanish with $\dt$ and is not addressed by that correction. The
comparison below therefore evaluates state construction and attenuation; it
is not presented as a formal errors-in-variables correction.

\section{Experiments and results}

The experiments move from controlled synthetic data to market proxies. Phase
1 checks the estimator when the variance state is observed. Phase 2 corrupts
that same state with measured amounts of noise. Phase 3 constructs variance
from OHLC prices, and Phases 4--7 test which conclusions survive across market
data, repeated simulations, a nonlinear alternative, and different proxies.
This ordering separates estimator error from state-construction error.

\subsection{Phase 1: multi-path latent-state recovery}

The first question is whether the fixed pipeline can recover the Heston
coefficients when the correct latent variance is available. We simulate the
model by Euler--Maruyama at an internal step of
$10^{-4}$ year and observe the latent log price and variance at daily
intervals for 100 years. This is the observation design of the original
notebook, with one correction: daily observation times are rounded on the
cumulative fine grid instead of truncating every day to 39 fine steps. This
removes the inherited clock drift without changing the estimator. The
generating parameters are
$(\kappa,\theta,\xi,\rho)=(2,0.04,0.3,-0.7)$. We repeat this exact
fine-step design for 30 independent seeds. Table~\ref{tab:baseline} reports
the resulting distribution rather than selecting one favorable path.

\begin{table}[H]
\centering
\caption{Latent-state recovery across 30 independent fine-step paths.}
\label{tab:baseline}
\small
\begin{tabular}{lrrrrr}
\toprule
Parameter & Truth & Median estimate & 10th--90th estimate & Median error\%\\
\midrule
$\kappa$ & 2.000 & 1.852 & $[1.540,2.262]$ & 12.29\% \\
$\theta$ & 0.0400 & 0.04034 & $[0.03653,0.04531]$ & 6.61\% \\
$\xi$ & 0.300 & 0.2973 & $[0.2946,0.3000]$ & 0.90\% \\
$\rho_{\mathrm{corr}}$ & $-0.700$ & $-0.7003$ & $[-0.7121,-0.6941]$ & 0.53\% \\
$\rho\xi$ & $-0.210$ & $-0.2064$ & $[-0.2135,-0.2040]$ & 1.80\% \\
\bottomrule
\end{tabular}
\end{table}

Every seed is below
5\% error for $\xi$, $\rho_{\mathrm{corr}}$, and $\rho\xi$, so the diffusion
and leverage result is stable across paths. The same threshold is reached in
only 7 of 30 fits for $\kappa$ and 10 of 30 for $\theta$; the corresponding
worst errors are 32.54\% and 21.46\%. This gap is exactly the pattern
predicted by Proposition~\ref{prop:efficiency}: the diffusion- and
leverage-type coefficients $\xi$, $\rho$, and $\rho\xi$ are governed by
in-fill consistency and improve with sampling frequency at fixed span,
whereas $\kappa$ and $\theta$ are governed by long-span consistency and are
limited by the 100-year horizon regardless of how finely the path is
sampled. Phase 1 therefore supports robust
latent-state diffusion and leverage recovery, not uniform sub-5\% recovery
of all Heston parameters. Figure~\ref{fig:baseline} retains seed 42 only to
show the fitted weak functions, while Fig.~\ref{fig:baseline-errors} displays
the full 30-seed error distributions.

\begin{figure}[H]
\centering
\includegraphics[width=\textwidth]{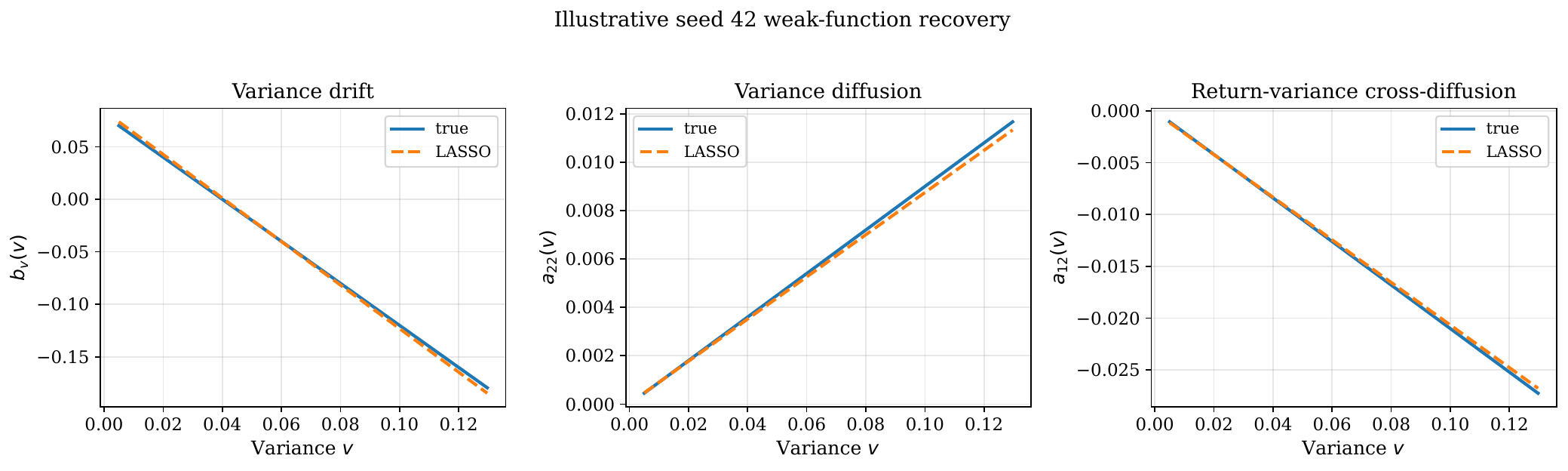}
\caption{Illustrative seed-42 weak-form fit. This path explains the recovered
functions but is not used as the multi-seed robustness summary.}
\label{fig:baseline}
\end{figure}

\begin{figure}[H]
\centering
\includegraphics[width=0.76\textwidth]{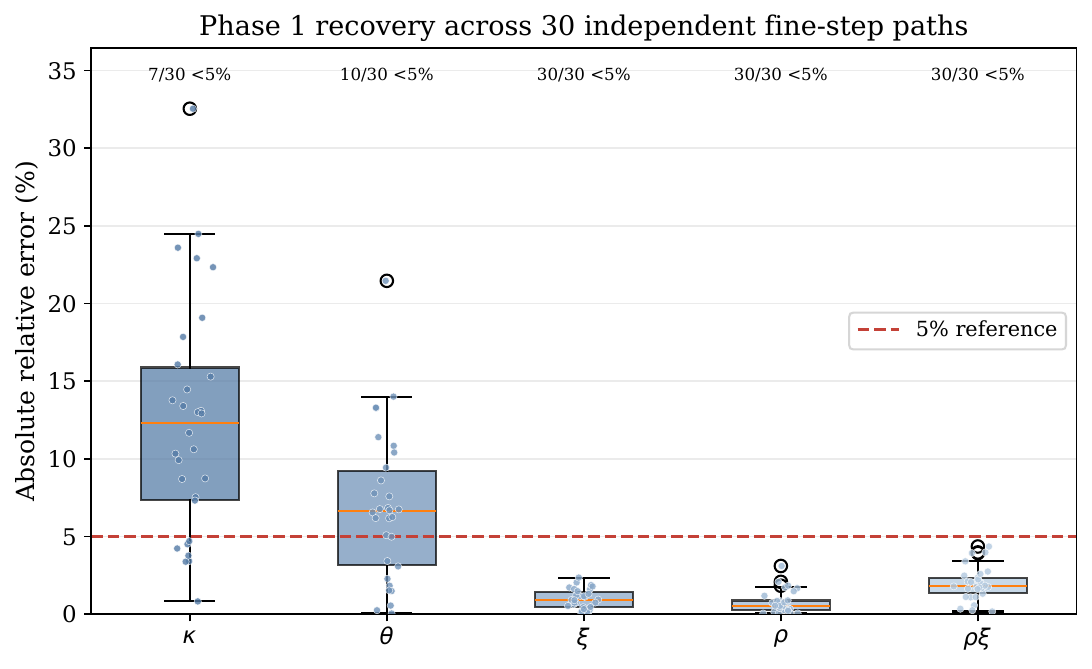}
\caption{Relative-error distributions across the 30 exact Phase 1 paths. The
5\% line separates the uniformly accurate diffusion and leverage quantities
from the less stable drift parameters; it is a descriptive reference, not a
general empirical acceptance threshold.}
\label{fig:baseline-errors}
\end{figure}

\subsubsection{Targeted kernel sensitivity}

Kernel placement controls which observations are averaged together. The
primary configuration remains 50 kernels and bandwidth factor 1.5. As a
targeted localization diagnostic, we additionally evaluate
$5\times5$ combinations of 25--100 kernels and bandwidth factors 0.75--3.0
using the illustrative seed-42 path and the same LASSO estimator. Across this grid, the relative error
of $\rho\xi$ remains between 1.81\% and 2.25\%. Drift recovery is more
sensitive: $\kappa$ error ranges from 0.38\% to 12.89\%, with wider kernels at
low kernel counts performing worst. The leverage result is therefore more
stable to localization choices than the drift estimate. The full grid in
Fig.~\ref{fig:kernel-sensitivity} shows that this path-specific localization
statement is not based only on the primary 50-kernel cell. It does not replace
the 30-seed variability summarized above.

\begin{figure}[H]
\centering
\includegraphics[width=0.92\textwidth]{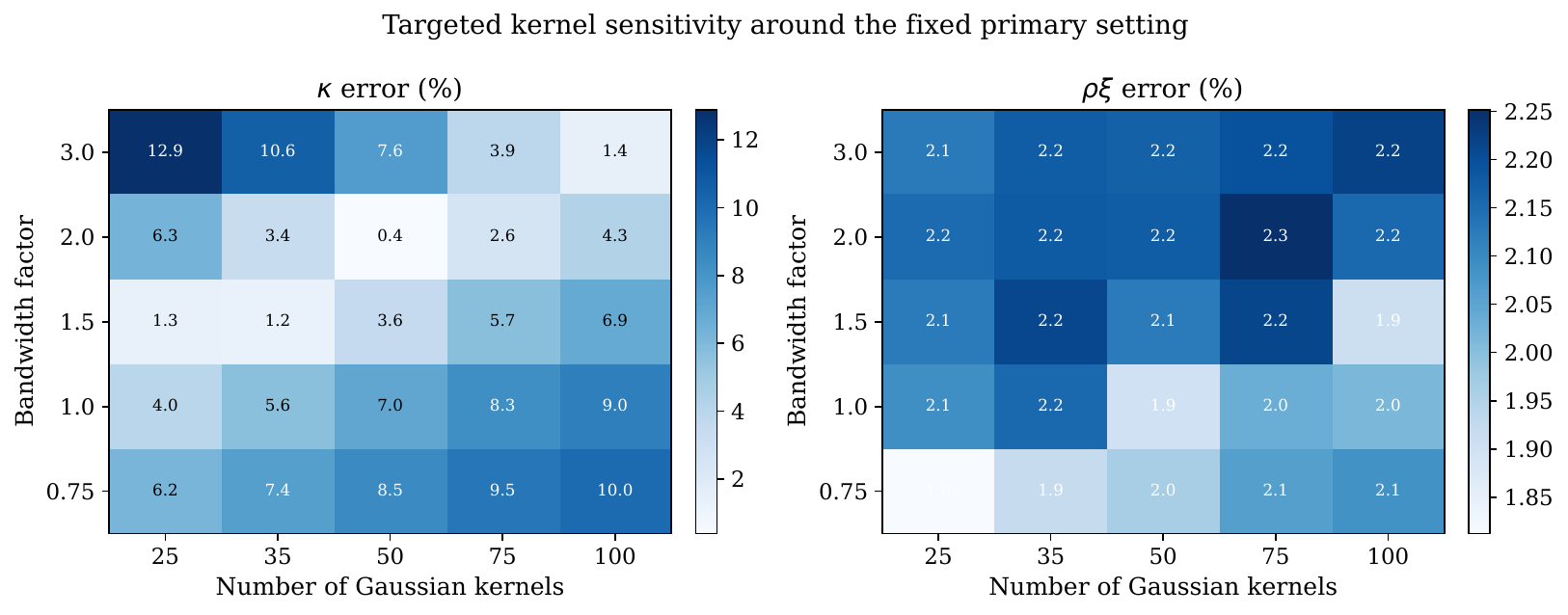}
\caption{LASSO sensitivity to kernel count and bandwidth. These are targeted
diagnostics around the fixed primary setting; they are not alternative
pipelines selected after seeing the empirical result.}
\label{fig:kernel-sensitivity}
\end{figure}
\FloatBarrier

\subsection{Phase 2: leverage robustness to state noise}

The next question is how leverage recovery degrades when only the state is
corrupted. The price path, true parameters, kernels, and LASSO pipeline are
kept fixed, so the experiment isolates state measurement error. The previous
presentation used one curve, emphasized $\kappa$, and described a threshold
near 32\%. We instead evaluate the quantity most directly tied to leverage,
$\rho\xi$. For each of 20 noise
levels, 100 independently generated Gaussian noise vectors are added to the
same reference variance path:
\begin{equation}
\widetilde v_n=\max\{v_n+\epsilon_n,10^{-12}\},\qquad
\epsilon_n\sim N(0,\gamma^2\,\mathrm{sd}(v)^2),
\end{equation}
where $\gamma$ ranges from 0 to 0.50. Thus the uncertainty band quantifies
measurement-noise draws conditional on one simulated path; it is not a
between-path confidence interval. Table~\ref{tab:noise} gives representative
grid points, and Fig.~\ref{fig:noise} shows all 20 levels and their dispersion.

\begin{table}[H]
\centering
\caption{Selected points from the replicated leverage-noise ablation. Errors
are relative absolute errors in $\rho\xi$; intervals are the 10th--90th
percentiles over 100 noise draws.}
\label{tab:noise}
\begin{tabular}{rrrr}
\toprule
Noise SD / state SD & Median error & 10th percentile & 90th percentile\\
\midrule
0.0\%  & 2.12\% & 2.12\% & 2.12\%\\
15.8\% & 4.78\% & 1.29\% & 8.16\%\\
31.6\% & 12.07\% & 5.70\% & 18.66\%\\
42.1\% & 17.37\% & 8.87\% & 26.98\%\\
50.0\% & 22.94\% & 11.66\% & 35.50\%\\
\bottomrule
\end{tabular}
\end{table}

The main result is gradual degradation rather than sudden failure. The
regenerated results do not establish a sharp phase transition, and the earlier
15\% horizontal reference had no external or decision-theoretic justification.
At 50\% injected
noise, a median 7.08\% of states are projected back to positivity; the
high-noise errors therefore include both measurement corruption and the
effect of this declared projection rule.

\begin{figure}[H]
\centering
\includegraphics[width=0.72\textwidth]{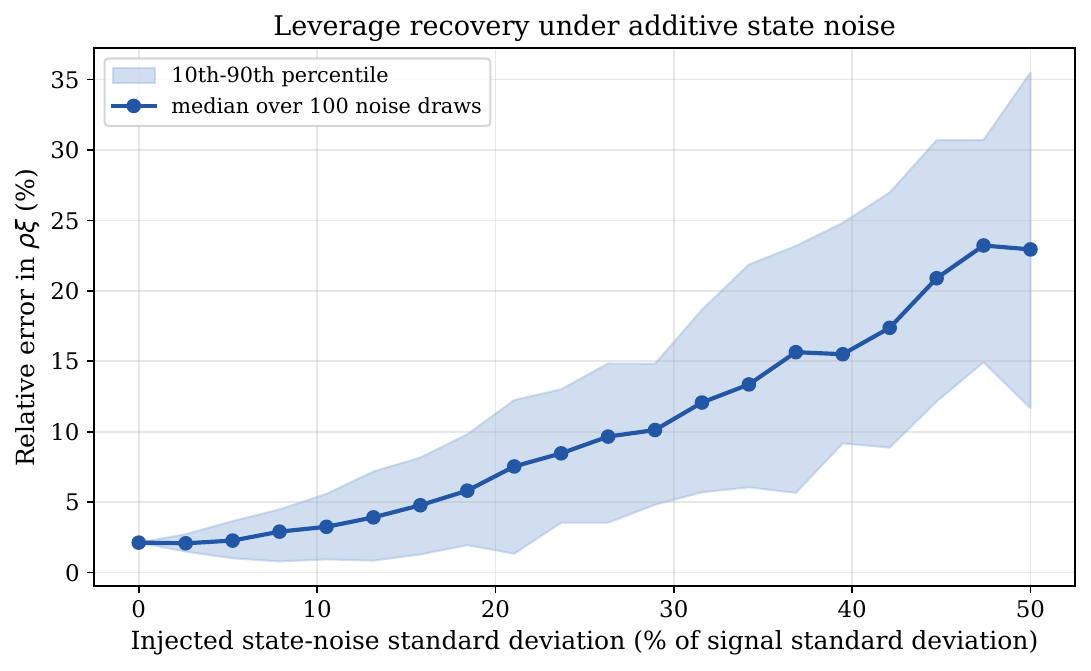}
\caption{Median leverage-coefficient error with a 10th--90th percentile band.
The plotted data and all 2,000 repeat-level estimates are supplied as CSV
files.}
\label{fig:noise}
\end{figure}
\FloatBarrier

\subsection{Phase 3: one-sided variance-proxy smoothing}

The noise ablation starts from an observed state, whereas market applications
must build the state from prices. This phase asks whether a one-sided filter
can produce a smoother daily variance series without sacrificing the
coefficients that we want to recover. Daily OHLC observations are generated
from one fixed-seed, 25,000-day Heston
simulation with 39 internal steps per day. The filter ranking is therefore a
within-path comparison rather than a cross-path guarantee. The raw
Garman--Klass variance proxy is
\begin{equation}
\widehat v^{GK}_n=\frac{1}{\dt}\left[
\frac12(\log H_n-\log L_n)^2-
(2\log2-1)(\log C_n-\log O_n)^2\right]
\cite{GarmanKlass1980}.
\end{equation}
We compare raw GK, EWMA spans 3--60, one-sided rolling means and medians,
outlier-clipped EWMA, log-EWMA, and local-level Kalman filters. The
predeclared switch rule replaces EWMA-14 only if an alternative has at least
20\% lower residual noise-to-signal ratio (NSR) and no larger relative error
in either $\kappa$ or $\rho\xi$.

Raw GK has residual NSR 0.759. EWMA-14 reduces it to 0.371 and raises its
correlation with latent variance to 0.952. The lowest NSR is 0.371 for
EWMA-14; the closest alternatives, Kalman $q/r=0.03$ (0.373) and EWMA-10
(0.377), do not achieve the required improvement. EWMA-14 is therefore
retained. Figure~\ref{fig:smoothing-state} shows how the retained filter
changes the simulated state.

\begin{figure}[H]
\centering
\includegraphics[width=\textwidth]{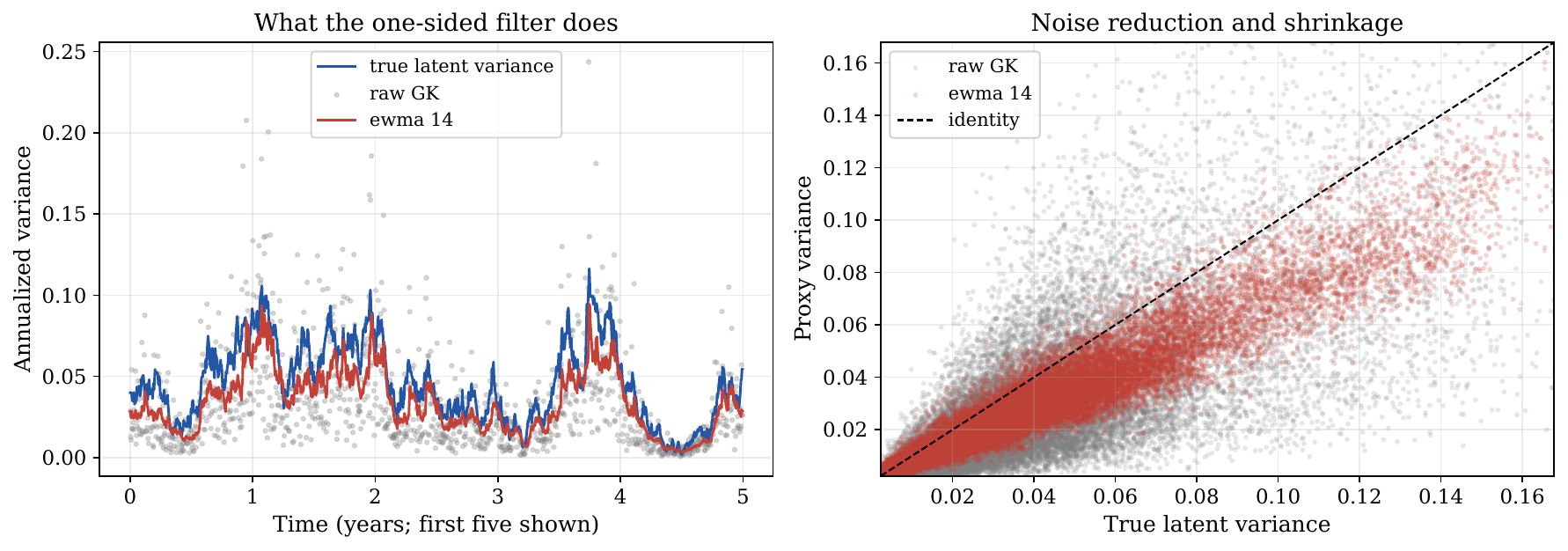}
\caption{Latent variance, raw GK, and the retained EWMA-14 proxy. The filter
is one-sided, and no time shift or future-looking smoother is used.}
\label{fig:smoothing-state}
\end{figure}

The trade-off is important. Smoothing lowers high-frequency proxy noise but
strongly attenuates the recovered leverage:
EWMA-14 gives $\widehat{\rho\xi}=-0.0280$ versus the true $-0.210$. Such
attenuation is consistent with the known difficulty of estimating leverage
from noisy or smoothed volatility proxies \cite{AitSahaliaFanLi2013}. Phase 3
therefore supports EWMA as a continuous, low-noise state construction. It
does not support calling the filter an EIV correction or a general solution
to proxy bias. Figure~\ref{fig:smoothing-tradeoff} places the noise reduction
and coefficient-recovery results together so that this limitation is visible.

\begin{figure}[H]
\centering
\includegraphics[width=\textwidth]{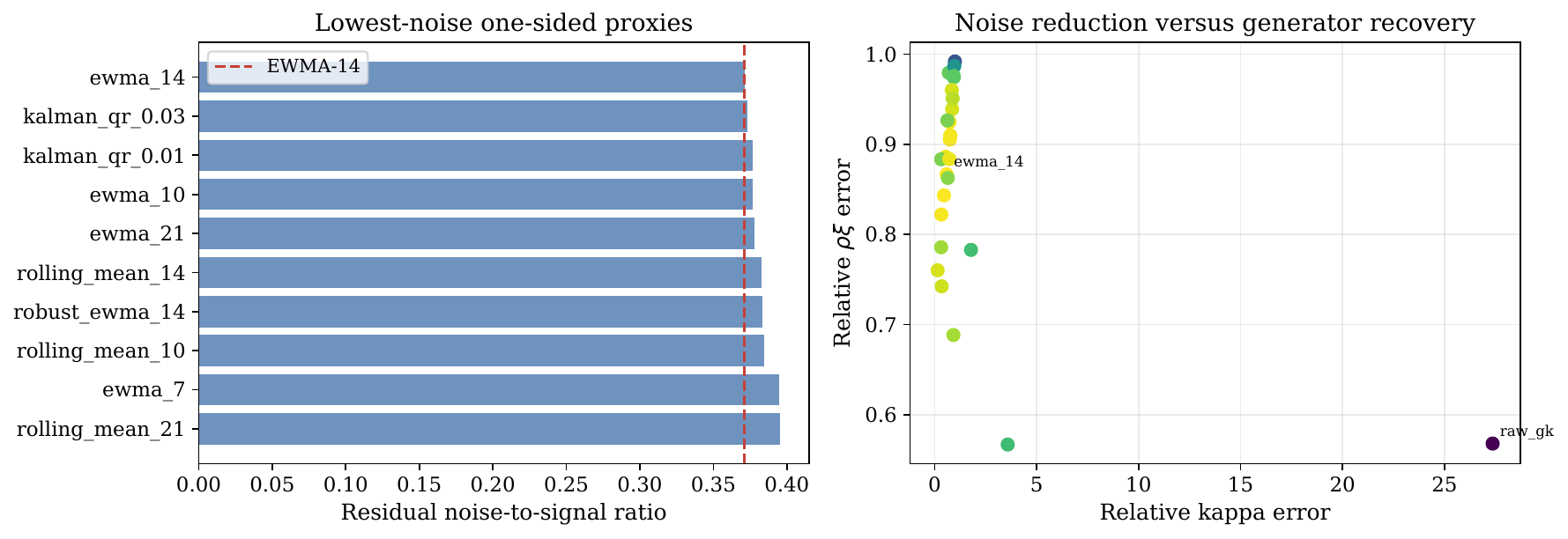}
\caption{One-sided proxy comparison showing the noise reduction and
parameter-recovery trade-off.}
\label{fig:smoothing-tradeoff}
\end{figure}
\FloatBarrier

\subsection{Phase 4: S\&P 500, 2007--2010}

We now apply the same state construction and estimator to an observed market
path. The goal is descriptive: we ask whether the fitted cross-diffusion is
negative during a period containing the financial crisis, not whether the
fit is a structural or option-pricing calibration. The empirical sample
contains 1,007 daily OHLC observations from 3 January
2007 through 30 December 2010. We construct same-day GK variance, apply
EWMA-14, and use the invariant unshifted LASSO pipeline.
Table~\ref{tab:sp500} reports the slope-derived estimates.

\begin{table}[H]
\centering
\caption{Contemporaneous S\&P 500 proxy-based estimates.}
\label{tab:sp500}
\small
\begin{tabular}{lrrrrrrr}
\toprule
Sample & $\widehat\kappa$ & $\widehat\theta$ & $\widehat\xi$ &
$\widehat a^{xx}_1$ & $\widehat{\rho\xi}$ &
$\widehat\rho_H$ & $\widehat\rho_{\mathrm{corr}}$\\
\midrule
2007--2010 & 2.361 & 0.0360 & 0.530 & 1.620 &
$-0.174$ & $-0.329$ & $-0.258$\\
\bottomrule
\end{tabular}
\end{table}

The long-run variance estimate corresponds to
$\sqrt{\widehat\theta}=18.96\%$ annualized physical/historical volatility. It
is not an implied-volatility estimate. The mean-reversion half-life
$\log(2)/\widehat\kappa$ is approximately 0.294 years, or 74 trading days.
The negative cross-diffusion is consistent with contemporaneous leverage, but
the absence of a latent variance ground truth prevents a claim that these are
the true Heston parameters. In particular,
$\widehat\rho_H=-0.329$ imposes $a^{xx}_1=1$, whereas the independently fitted
price-diffusion slope is 1.620; using that slope gives
$\widehat\rho_{\mathrm{corr}}=-0.258$. The fitted slopes also fail the Feller
condition:
\begin{equation}
\frac{2\widehat\kappa\widehat\theta}{\widehat\xi^2}=0.604<1.
\end{equation}
Moreover, the affine diffusion fits retain nonzero constants
\begin{equation}
(\widehat a^{vv}_0,\widehat a^{xx}_0,\widehat a^{xv}_0)
=(0.00121,-0.00197,0.00156).
\end{equation}
The slope-derived quantities should therefore
be read as summaries of an affine Heston-type proxy fit, not as an admissible
exact Heston calibration. Figure~\ref{fig:sp500} displays the three fitted
generator components, while Fig.~\ref{fig:sp500-state} shows the time-domain
proxy and open-to-open returns that enter those weak equations.

\begin{figure}[H]
\centering
\includegraphics[width=\textwidth]{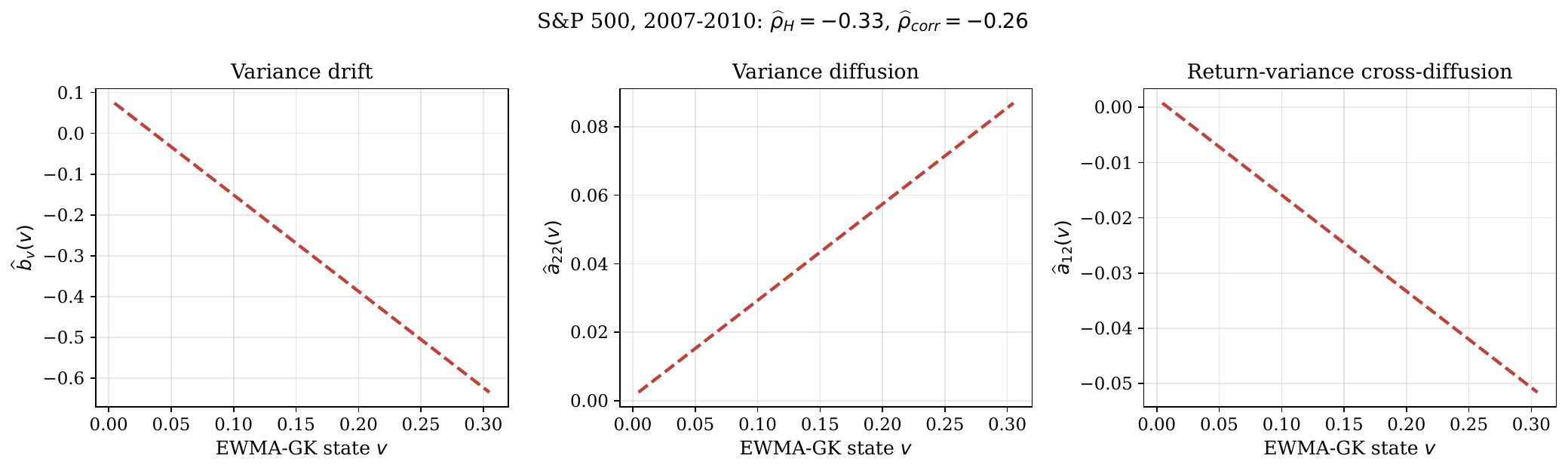}
\caption{Contemporaneous weak-form functions for the S\&P 500 sample.}
\label{fig:sp500}
\end{figure}

\begin{figure}[H]
\centering
\includegraphics[width=0.92\textwidth]{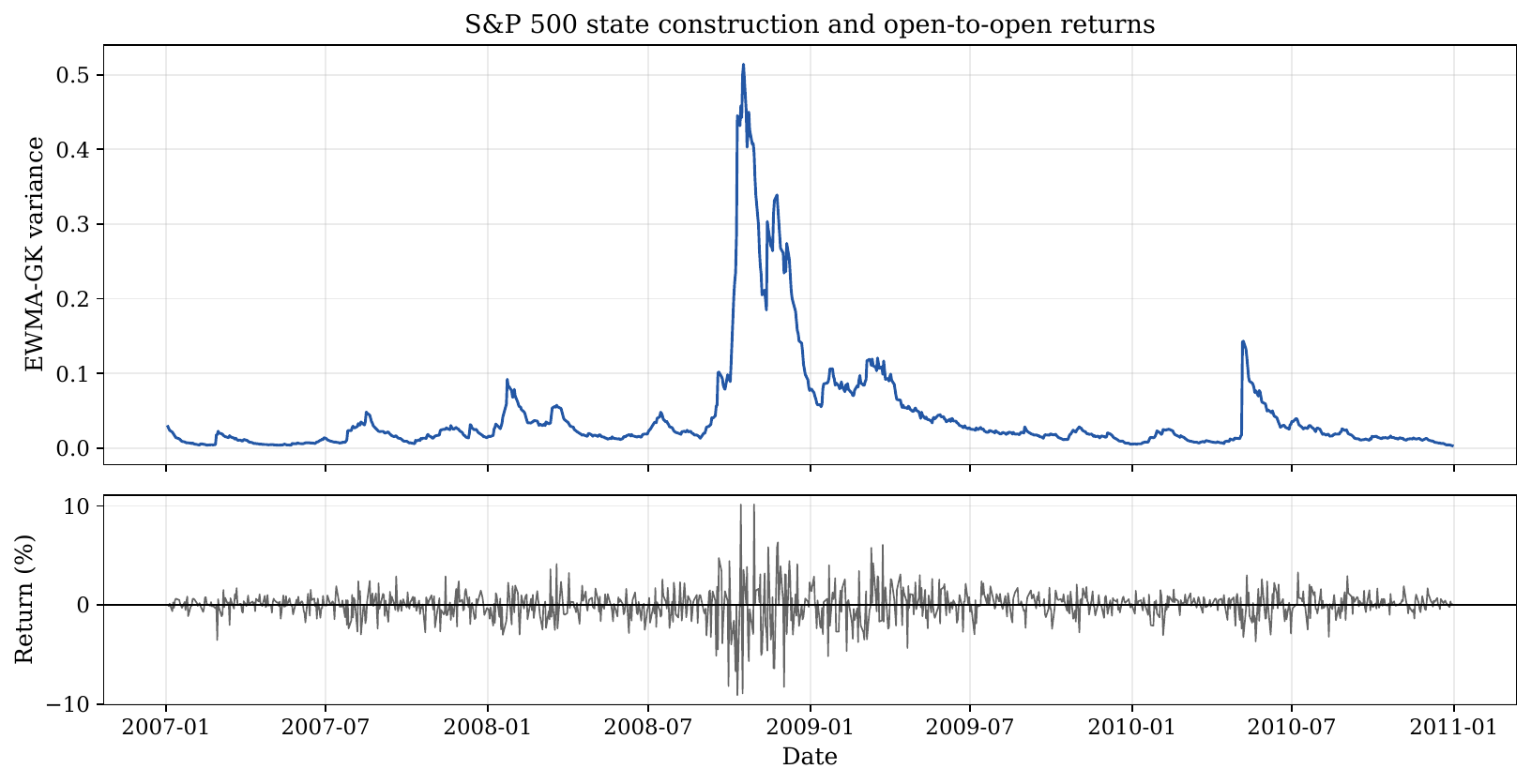}
\caption{The empirical state construction in time. The variance proxy rises
during the crisis while large daily returns cluster. This plot explains the
input to the generator fit; it is not a fitted-price or forecasting result.}
\label{fig:sp500-state}
\end{figure}
\FloatBarrier

\subsection{Phase 5: multi-seed leverage regimes}

One empirical estimate cannot show whether leverage recovery is systematic.
We therefore return to latent-state simulation and generate 30 paths at each
of four correlations. Every other generating and fitting choice is held
constant. Table~\ref{tab:rho} summarizes the recovered distributions;
Fig.~\ref{fig:rho} shows both their alignment with the identity line and the
absolute-error spread.

\begin{table}[H]
\centering
\caption{Leverage recovery across 30 seeds. Intervals are the 10th--90th
percentiles of $\widehat\rho$.}
\label{tab:rho}
\begin{tabular}{rrrrr}
\toprule
True $\rho$ & Median $\widehat\rho$ & 10th & 90th &
Median absolute error\\
\midrule
$-0.20$ & $-0.200$ & $-0.216$ & $-0.182$ & 0.0079\\
$-0.50$ & $-0.503$ & $-0.515$ & $-0.490$ & 0.0086\\
$-0.65$ & $-0.651$ & $-0.658$ & $-0.639$ & 0.0054\\
$-0.85$ & $-0.850$ & $-0.855$ & $-0.844$ & 0.0033\\
\bottomrule
\end{tabular}
\end{table}

All 120 latent-state fits recover the correct sign. The result shows that the
cross-variation target works across weak and strong negative leverage when
the model is matched and the state is observed. It does not remove the
empirical proxy problem.

\begin{figure}[H]
\centering
\begin{subfigure}{0.49\textwidth}
\includegraphics[width=\textwidth]{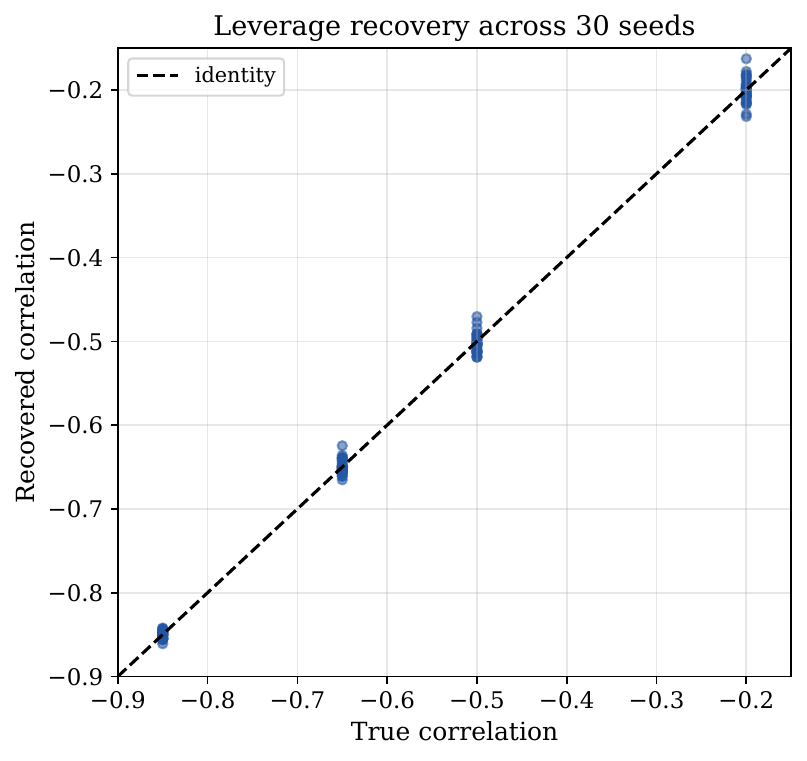}
\end{subfigure}
\begin{subfigure}{0.49\textwidth}
\includegraphics[width=\textwidth]{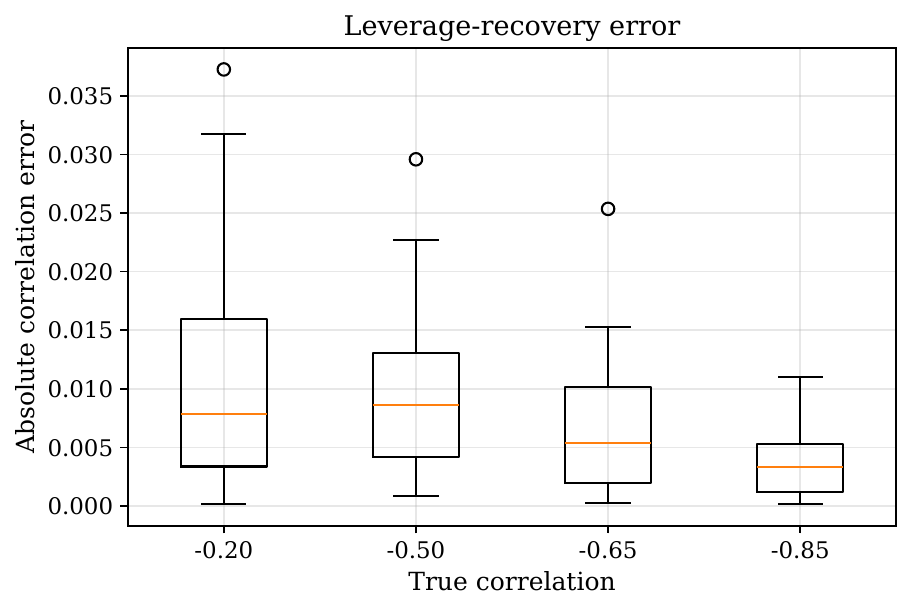}
\end{subfigure}
\caption{Multi-seed latent-state leverage recovery.}
\label{fig:rho}
\end{figure}
\FloatBarrier

\subsection{Phase 6: nonlinear-drift falsification}

Good recovery inside a linear Heston library does not show that LASSO can
discover a larger nonlinear model. We test that stronger claim directly by
adding a quadratic term to the simulated variance drift,
\begin{equation}
b_v(v)=1.5(0.04-v)-\beta(v-0.04)^2,
\end{equation}
and extend the LASSO library to $[1,v,v^2]$. A useful structure-discovery
method should rarely select $v^2$ when $\beta=0$. Here it is selected in
28 of 30 null simulations (93.3\%). At $\beta=15$, the median fitted
quadratic coefficient is $-12.18$, but the 10th--90th percentile interval
$[-23.96,0.93]$ still crosses zero. Selection rates of 93--97\% across all
tested values are therefore dominated by false positives. ``Selected'' here
means an absolute fitted coefficient greater than the code tolerance
$10^{-12}$; no post-LASSO significance test is applied.
The useful conclusion is negative: nonzero support is not reliable evidence
of a new drift term under this configuration. Figure~\ref{fig:nonlinear}
makes the failure visible in both selection frequency and fitted-coefficient
spread.

\begin{figure}[H]
\centering
\includegraphics[width=0.68\textwidth]{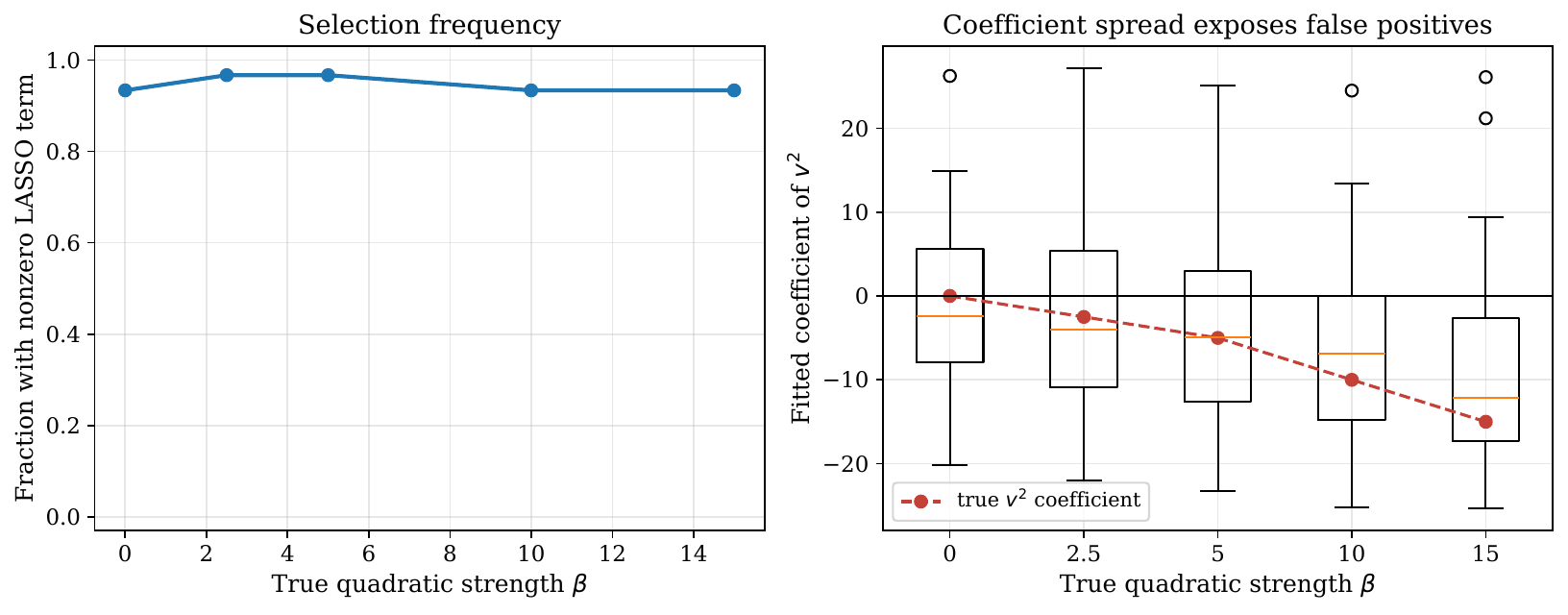}
\caption{Quadratic selection frequency. The high null selection rate
falsifies a reliable nonlinear-support-discovery claim for this configuration.}
\label{fig:nonlinear}
\end{figure}
\FloatBarrier

\subsection{Phase 7: Indian-equity proxy sensitivity}

The final phase asks a practical question that cannot be answered with one
variance proxy: does the sign of recovered leverage remain negative when the
same companies are analyzed using different state constructions? The stress
test uses a balanced panel of 50 Indian equities from
1 February through 30 June 2020 and five contemporaneous variance proxies:
one-minute realized variance, raw GK, EWMA-GK, Parkinson, and Yang--Zhang
\cite{AndersenEtAl2003,Parkinson1980,YangZhang2000}. Each company--proxy pair
is processed by the same LASSO pipeline. The panel has 100 retained trading
days per company (5,000 company-days); every retained company-day contains
316--376 valid one-minute rows. The supplied one-minute Parquet archive is
aggregated upstream to the daily proxy panel, and the standardized recovery
script starts from that panel.

\begin{table}[H]
\centering
\caption{Proxy sensitivity of the recovered leverage sign across 50
companies. Parameter magnitudes are not treated as structural estimates in
this short, noisy panel.}
\label{tab:india}
\begin{tabular}{lrrr}
\toprule
Variance proxy & Valid $N$ & Median $\widehat\rho_{\mathrm{corr}}$ &
Fraction $\widehat\rho_{\mathrm{corr}}<0$\\
\midrule
One-minute realized variance & 50 & $-0.320$ & 92\%\\
Raw Garman--Klass & 49 & $-0.424$ & 90\%\\
Parkinson & 50 & $-0.229$ & 80\%\\
EWMA Garman--Klass & 41 & $-0.204$ & 58\%\\
Yang--Zhang & 50 & $0.000$ & 38\%\\
\bottomrule
\end{tabular}
\end{table}

The comparison gives a mixed but interpretable result. Negative leverage is
common for the realized-variance, GK, and Parkinson proxies, but it is not
universal. Smoothing weakens the sign rate, and
Yang--Zhang does not support a panel-wide negative result. The appropriate
conclusion is therefore \emph{proxy-sensitive sign evidence}, not
proxy-robust identification. Drift and diffusion magnitudes are unstable in
this panel---several fits imply nonpositive mean reversion---so they are not
given a structural interpretation. ``Valid $N$'' counts finite
covariance-normalized correlations; the sign-rate denominator remains all 50
companies, so an undefined correlation is not counted as negative.
Table~\ref{tab:india} supplies the proxy-by-proxy numbers.
Figure~\ref{fig:india-generator} retains the failed realized-variance
generator as a diagnostic rather than hiding it.

\begin{figure}[H]
\centering
\includegraphics[width=0.92\textwidth]{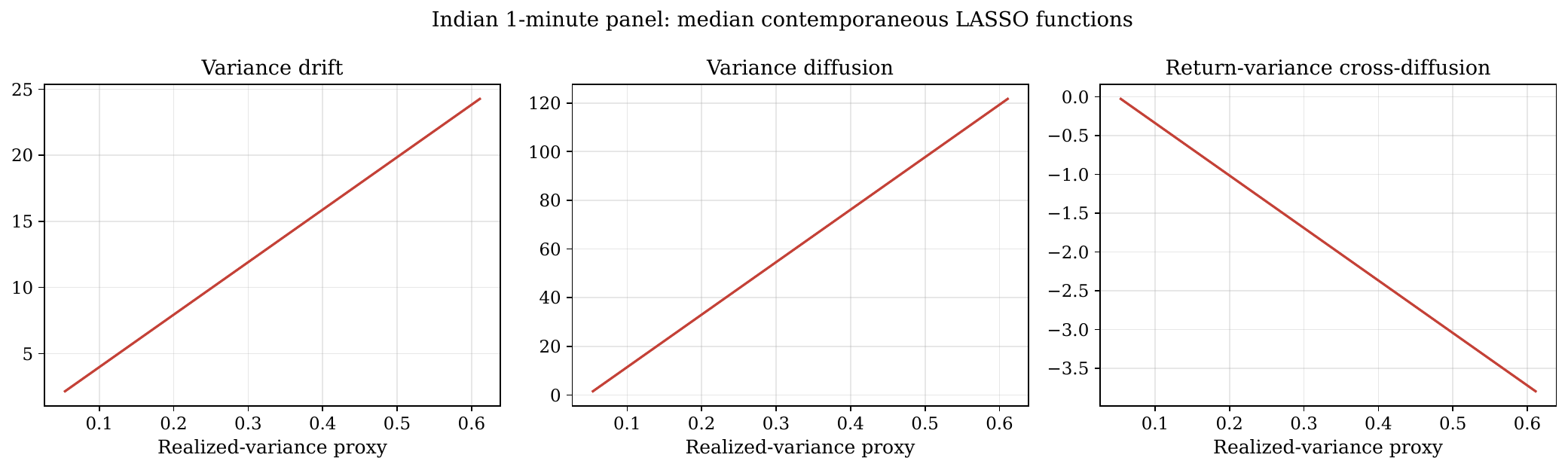}
\caption{Median fitted drift, variance-diffusion, and cross-diffusion
functions for the one-minute realized-variance proxy. The non-mean-reverting
median drift is displayed as a failure diagnostic and is not interpreted as
a valid Heston generator.}
\label{fig:india-generator}
\end{figure}

Figures~\ref{fig:india} and \ref{fig:india-sign} separate the distribution of
finite covariance-normalized correlations from the simpler all-company sign
rates. This distinction matters because a zero or undefined LASSO slope is
not evidence of a negative relationship.

\begin{figure}[H]
\centering
\includegraphics[width=0.60\textwidth]{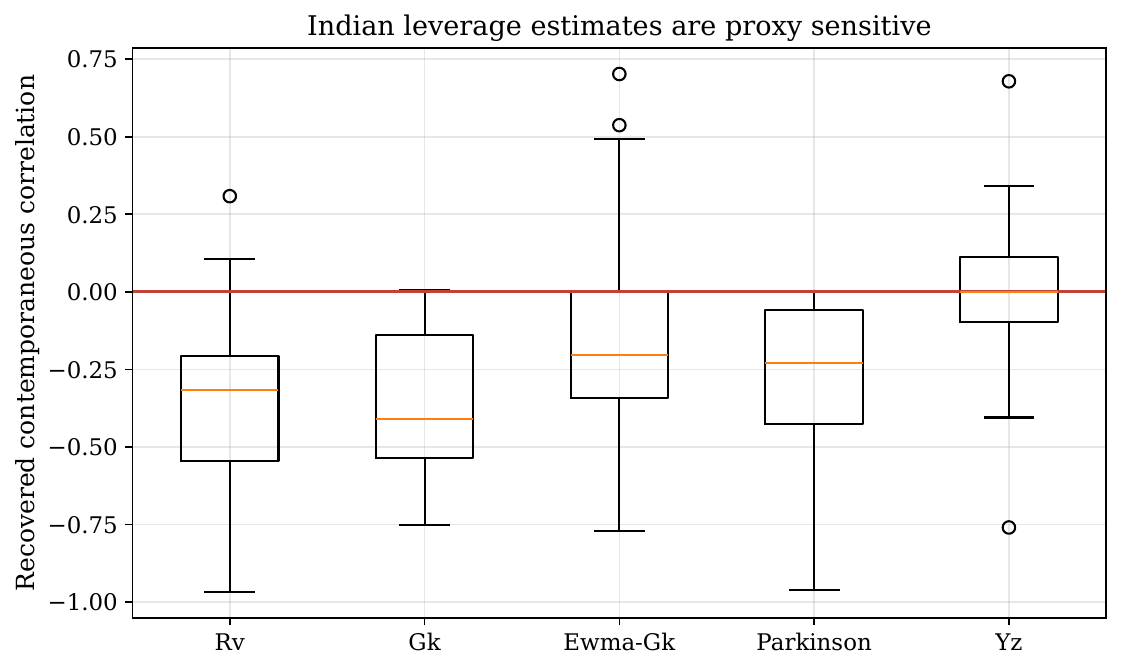}
\caption{Distribution of admissible contemporaneous correlation estimates by
variance proxy.}
\label{fig:india}
\end{figure}

\begin{figure}[H]
\centering
\includegraphics[width=0.72\textwidth]{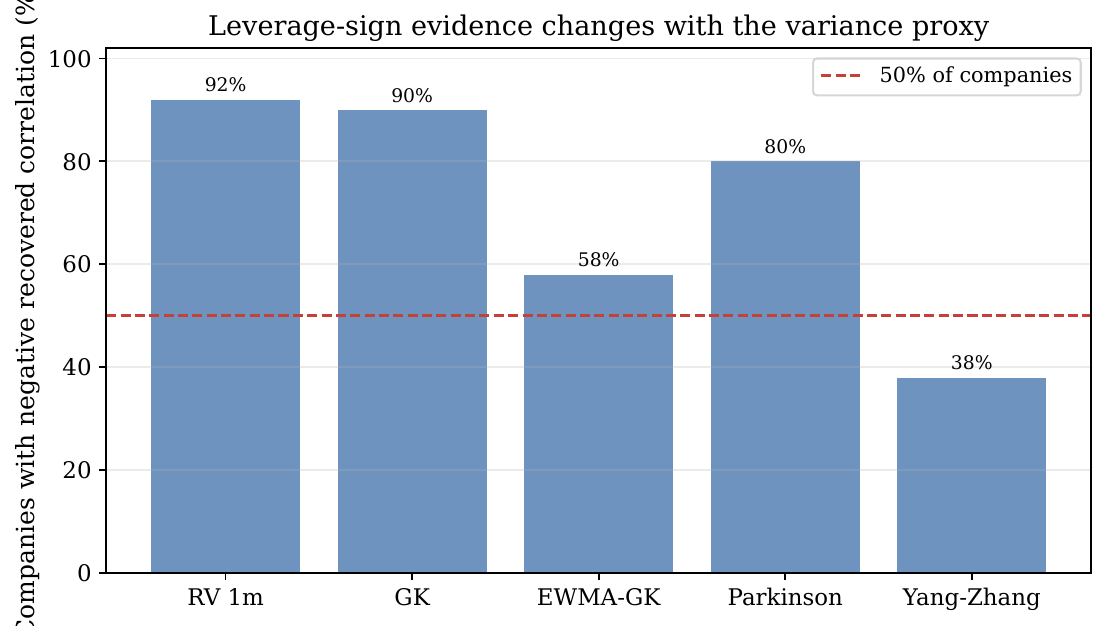}
\caption{The proxy-dependent conclusion stated directly as a sign rate. Raw
realized-variance and range estimators support negative leverage more often
than the smoothed or Yang--Zhang proxies.}
\label{fig:india-sign}
\end{figure}
\FloatBarrier

\section{Discussion}

The experiments are designed to separate three cases: recovery with an
observed state, recovery after controlled state corruption, and recovery from
market proxies. The evidence is strong in the first case and progressively
more conditional in the other two. This distinction is the main guide for
interpreting the results.

\subsection{What the experiments support}

With an observed latent state and a prespecified Heston library, weak-form
LASSO accurately recovers variance diffusion, cross-diffusion, and the
leverage correlation. In the 30 exact fine-step Phase 1 paths, all estimates
of $\xi$, $\rho$, and $\rho\xi$ remain below 5\% error. Drift is less stable,
with median errors of 12.29\% for $\kappa$ and 6.61\% for $\theta$. The
separate leverage-regime experiment recovers the correct sign in 120
additional latent-state simulations. The same machinery also reproduces the reported
contemporaneous S\&P leverage ratio $\widehat\rho_H=-0.329$ under an explicit
Heston normalization; the corresponding fitted covariance correlation is
$\widehat\rho_{\mathrm{corr}}=-0.258$.

The results do not justify stronger claims. Additive state noise gradually
degrades leverage recovery; no sharp threshold is established. One-sided
smoothing creates a continuous state with lower proxy noise but can attenuate
cross-variation substantially. The empirical leverage sign changes with the
proxy. Finally, the nonlinear null experiment shows that nonzero LASSO support
alone is insufficient evidence of a newly discovered drift term.

\subsection{Limitations}

The limitations below are not separate from the results; they explain why the
synthetic and empirical claims are stated differently.

\begin{itemize}[leftmargin=*]
  \item \textbf{Contemporaneous state construction.} Same-day OHLC uses
  information accumulated during the day. The empirical estimator is
  descriptive and cannot be interpreted as an ex ante trading signal.
  \item \textbf{Generated regressor.} The variance proxy is noisy and
  smoothed. Standard cross-validation does not correct errors in variables,
  and coefficient uncertainty from state construction is not propagated.
  \item \textbf{Overlapping weak equations.} The five cross-validation folds
  split 50 kernel equations whose underlying trajectory contributions
  overlap. The folds are therefore tuning folds, not independent
  out-of-sample observations; this is especially relevant to the nonlinear
  false-selection result.
  \item \textbf{One market path.} The S\&P sample contains a crisis and only
  1,007 daily observations. There is no latent-state truth against which to
  validate its parameter estimates.
  \item \textbf{Return timing.} The S\&P open-to-open return contains an
  overnight component, whereas same-day Garman--Klass variance is based on
  the intraday high, low, open, and close. The retained contemporaneous
  convention does not eliminate this measurement mismatch.
  \item \textbf{Indian sample selection.} The balanced 50-company panel was
  selected using coverage, liquidity, and sector information from the same
  stress window. It is a deliberately balanced stress sample; it is neither
  capitalization weighted nor demonstrably free of survivorship bias.
  \item \textbf{Restricted library.} The main result estimates coefficients
  inside the Heston form; it does not discover that form without prior
  specification.
  \item \textbf{Simulation discretization.} Phase 1 uses a fine internal step
  before daily observation for all 30 paths, whereas the separate leverage
  regime screen uses direct daily full-truncation Euler for computational
  replication. The estimator is invariant, but coefficient accuracy should
  not be compared across these observation designs without accounting for
  discretization.
  \item \textbf{Uncertainty.} The reported percentile bands describe
  simulation or injected-noise repetitions, not formal confidence intervals
  for the empirical generator.
\end{itemize}

\subsection{Next methodological steps}

The most important next step is to treat the variance state and the generator
as a joint estimation problem. Possible approaches include latent-state
filtering, instrumental-variable weak moments, and uncertainty propagation
through the proxy construction. The selection problem also needs a stronger
validation design, such as separate tuning and evaluation paths or stability
selection calibrated under a null model. For higher-frequency data,
microstructure noise and asynchronous sampling should be modeled explicitly
rather than left to a smoother.

\section{Conclusion}

The central result is specific rather than uniform. Across 30 synchronized
fine-step, daily-observed paths, the standardized weak-form LASSO pipeline
keeps $\xi$, $\rho$, and $\rho\xi$ below 5\% error in every seed. It does not
provide the same guarantee for drift: median errors are 12.29\% for $\kappa$
and 6.61\% for $\theta$. The separate leverage-regime experiment recovers the
correct leverage sign in all 120 paths.

The empirical conclusion is narrower. The contemporaneous S\&P 500 fit gives
$\widehat\rho_H=-0.329$ and
$\widehat\rho_{\mathrm{corr}}=-0.258$, but the Indian panel shows that sign
frequency depends materially on the variance proxy. The replicated noise
ablation explains part of this sensitivity through a gradual error profile,
and the smoothing comparison shows that lower proxy noise can come with
strong leverage attenuation. The nonlinear falsification further limits the
claim to coefficient recovery in a prespecified Heston library. In short, the
method identifies leverage well under controlled conditions and provides
proxy-sensitive contemporaneous evidence in market data; it does not solve
latent-state estimation or general stochastic-structure discovery.

\section*{Data and code availability}

Code and supplementary materials for this work are available at \url{https://github.com/Sathvik-Gullipalli/Leverage-recovery-Weak-SINDy}. The repository will be updated with the final reproducibility package upon publication.

\section*{Indian panel composition}

The balanced panel contains five selected names from each of ten broad
sectors. This table restores the sample definition from the previous draft;
it does not imply that sector balance makes the 50 firms representative of
the full Indian equity market.

\begin{table}[H]
\centering
\scriptsize
\begin{tabular}{p{0.31\textwidth}p{0.63\textwidth}}
\toprule
Sector & Symbols\\
\midrule
Energy and utilities & RELIANCE, POWERGRID, ONGC, NTPC, COALINDIA\\
Non-bank financials and insurance & BAJFINANCE, BAJAJFINSV, HDFCLIFE, CHOLAFIN, SBILIFE\\
Banks & HDFCBANK, ICICIBANK, AXISBANK, SBIN, KOTAKBANK\\
Consumer discretionary, services, telecom & BHARTIARTL, ASIANPAINT, TITAN, INDIGO, ADANIPORTS\\
Automobiles & MARUTI, TMPV, M\&M, BAJAJ-AUTO, EICHERMOT\\
Consumer staples & HINDUNILVR, ITC, NESTLEIND, BRITANNIA, DABUR\\
Information technology & TCS, INFY, HCLTECH, TECHM, WIPRO\\
Industrials, capital goods, construction & LT, BEL, SIEMENS, ABB, HAL\\
Healthcare & SUNPHARMA, CIPLA, DRREDDY, DIVISLAB, APOLLOHOSP\\
Materials, metals, cement & TATASTEEL, ULTRACEMCO, JSWSTEEL, HINDALCO, GRASIM\\
\bottomrule
\end{tabular}
\end{table}

\section*{Acknowledgments}

This research was conducted at the Quantum and Nano Devices (QuaNaD) Lab, PES University. The authors would like would like to acknowledge the support of the QuaNaD Lab and thank Prof. Shreyus, Mayank, Omkar, and Tabassum for being part of the research environment during my internship.

\end{document}